\documentclass[11pt]{article}

\usepackage[utf8]{inputenc}
\usepackage[colorlinks=false,hidelinks]{hyperref}
\usepackage[square,sort&compress,comma,numbers]{natbib}
\usepackage{amsmath,amsfonts,amssymb,amsthm,mathrsfs}
\usepackage{braket}
\usepackage{amsthm}
\usepackage{dsfont}
\usepackage{color}
\usepackage[title]{appendix}
\usepackage{setspace}
\usepackage{soul}

\newtheorem{vor}{Assumption}[section]
\newtheorem{theorem}[vor]{Theorem}

\newtheorem{lemma}[vor]{Lemma}

\theoremstyle{definition}

\numberwithin{equation}{section}

\renewcommand{\rm}{\textnormal}

\newcommand{\mb}{\mathbb}

\newcommand{\mf}{\mathfrak}
\newcommand{\mc}{\mathcal}
\newcommand{\Ls}{\big\langle}
\newcommand{\Rs}{\big\rangle}
\newcommand{\ls}{\langle}
\newcommand{\rs}{\rangle}

\renewcommand{\d}{\textnormal{d}}
\renewcommand{\l}{\ell}
\newcommand{\el}{\ell}

\begin{document}

\title{\huge{Exponential decay of the number of excitations in\\ the weakly interacting Bose gas}
}

\author{David Mitrouskas\thanks{Institute of Science and Technology Austria (ISTA), Am Campus 1, 3400 Klosterneuburg, Austria. \texttt{mitrouskas@ist.ac.at}}
\phantom{i} and Peter Pickl\thanks{Universität Tübingen, Fachbereich Mathematik, Auf der Morgenstelle 10, 72076 Tübingen, Germany. \texttt{p.pickl@uni-tuebingen.de}}
}

\date{\today}
\maketitle

\frenchspacing

\begin{spacing}{1.1}

\begin{abstract} 
We consider $N$ trapped bosons in the mean-field limit with coupling constant $\lambda_N=1 / (N-1)$. The ground state of such systems exhibits Bose--Einstein condensation. We prove that the probability of finding $\ell$ particles outside the condensate wave function decays exponentially in $\ell $.
\end{abstract}

\section{Introduction and main result}

We consider $N$ bosons described by the Hamiltonian
\begin{align}
H_N = \sum_{i=1}^N \Big( -\Delta_i + V^{\rm{ext}}(x_i) \Big) +\frac{1}{N-1} \sum_{1\le i < j \le N } v(x_i-x_j)
\end{align}
acting on the $N$-particle Hilbert space $\mathfrak H^N = \bigotimes^N_{\textnormal{sym}} L^2(\mathbb R^3)$. We assume that the external potential $V^{\rm{ext}} :\mb R^3 \to \mb R$ is measurable, locally bounded and acts as a confining potential, i.e.,  $V^{\textnormal{ext}}(x) \to \infty$ as $|x|\to \infty$. For the pair potential, we consider two cases: (i) $v$ is an even, pointwise bounded function with non-negative Fourier transform or (ii) $v(x) = \lambda |x|^{-1}$ with $\lambda >0$ the repulsive Coulomb potential. Under these conditions, $H_N$ is essentially self-adjoint and has a unique ground state, which we denote by $\Psi_N$. It is well known \cite{LNR,LNSS,GrechS,Seiringer} that the ground state exhibits complete Bose--Einstein condensation (BEC) in the minimizer $\varphi \in L^2(\mathbb R^3)$ of the Hartree energy functional $u \mapsto  N^{-1} \langle u^{\otimes N}, H_N u^{\otimes N}\rangle$ (for BEC in dilute Bose gases, see \cite{NT,BSS,BBCS20,BBCS19,LSSY,LS02,NNR,NRS}). Complete BEC means that in the limit $N\to \infty$, all but finitely many particles occupy the same one-particle state $\varphi \in L^2(\mb R^3)$. To make this statement more precise, let $p = |\varphi \rangle\langle \varphi |$ and consider the family of operators  $ 
\mf P_{N}(\ell ) : \mathfrak H^N \to \mathfrak H^N $ with $\ell \in \{0,\ldots, N\}$
\begin{align}
\mf P_N(\ell) = \Big(   ( 1 -p ) ^{\otimes \ell} \otimes p ^{\otimes N-\ell} \Big)_{\rm{sym}}.
\end{align}
It is straightforward to verify that 
\begin{align}
\mf P_N(\ell) \mf P_N(k) = \delta_{\ell k}\quad \text{and}\quad \mathds 1 = \sum_{\ell=0}^N \mf P_N(\ell).
\end{align}
The operator $\mf P_N(\ell)$ projects onto states that contain exactly $N-\ell$ particles in the condensate wave function $\varphi$ and $\ell$ particles in the orthogonal complement $\{ \varphi \}^\perp \subseteq L^2(\mb R^3)$.  The number $P_N(\ell ) : = \| \mf P_{N}(\ell ) \Psi_N \|^2 $ is thus the probability of finding $\ell $ particles in the ground state $\Psi_N$ that do not occupy the condensate wave function $\varphi$. Complete BEC can now be formulated as $P_N(0) = 1 +  O(N^{-1})$ as $N\to \infty$. For finite values of $N$, however, there is in general a non-vanishing probability $1-P_N(0) = \sum_{\ell=1}^N P_N(\ell) >0$ of finding particles outside the condensate. This work aims to establish strong asymptotic bounds on $P_N(\ell)$ for large $\ell$ and $N$. Our main theorem states that $P_N(\ell)$ decays exponentially in $\ell$.

\begin{theorem} \label{theorem}There exists an $\varepsilon>0$ such that for every  $ f :\mb N  \to \mb N$ with $f(n) \xrightarrow{n \to \infty} \infty$ 
\begin{align}\label{eq:main:bound}
\lim_{N\to \infty} P_N(f(N)) \, \exp( \varepsilon f(N) )  \ = \ 0. 
\end{align}
\end{theorem}

For the homogeneous Bose gas on the torus with bounded pair potential of positive type, the theorem was stated and proved already in 2017 within the PhD thesis \cite[Theorem 3.1]{Mitrouskas2017}. In this note, we present essentially the same proof, albeit somewhat simplified and with the correction of two minor errors. The generalizations to the trapped Bose gas and the repulsive Coulomb potential require only small modifications. As shown in \cite{Mitrouskas2017}, the statement can be proved also for excited eigenstates of $H_N$ but we shall address only the ground state here. 

A related result was obtained in \cite{BPS2021}, where the authors derive, under similar assumptions on $v$, higher-moment bounds of the form $\sum_{\ell=1}^N \ell^n P_N(\ell) \le C_n$ for all $n\in \mb N$ with unspecified constants $C_n$ (if $v$ is bounded, they obtain $C_n \le (C (n+1))^{(n+6)^2}$ for some $C>0$). Theorem \ref{theorem} implies that $C_n \le C^n n!$ for some $C>0$. In \cite{BPS2021}, the higher-moment bounds are used to derive an asymptotic series for the ground state energy of $H_N$. Our improved bounds on $C_n$ could thus be relevant for the analysis of certain resummation properties, such as Borel summability, of this asymptotic series, see \cite[Remark 3.5]{BPS2021}.

Very recently, Nam and Rademacher \cite{NR23} achieved a major advancement by extending the exponential decay of $P_N(\ell)$ to dilute Bose gases. They consider the homogeneous Bose gas on the unit torus with pair potential $v(x) = N^{3\beta} v(N^\beta x)$, $\beta\in [0,1]$, for non-negative compactly supported $v\in L^3([0,1])$. This includes, in particular, the physically most relevant Gross--Pitaevskii regime ($\beta=1$). Their result shows that for every low-energy eigenfunction $\psi_N \in \mf H^N$, $ \ls \psi_N, e^{\kappa \mathcal N} \psi_N \rs = O(1)$ as $N\to \infty$ for some $\kappa>0$, where $\mathcal N = \sum_{i=1}^N (1-p_i)$ is the operator that counts the number of particles outside the condensate. Since $\ls \psi_N, e^{\kappa \mathcal N} \psi_N \rs  = \sum_{\ell =1}^N P_N(\ell) \exp(\kappa \ell)$, this proves \eqref{eq:main:bound} for the dilute Bose gas. Higher-moment bounds of the form $\ls \psi_N , \mathcal N^n \psi_N \rs \le C_n$, $n\in \mb N$, have been obtained in the Gross--Pitaevskii regime in \cite{BBCS19}.

Finally, let us mention that exponential bounds for slightly different observables than $\mathcal N$ have been recently studied also in the context of large deviations \cite{KRS,Rademacher,RademacherS} (see also \cite[Remark 1.3]{NR23}).

The idea of the proof of Theorem \ref{theorem} is to show that the function $P_N(\ell)$ satisfies an inequality of the form $  P_N(\ell+2) +  P_N(\ell-2) - 2 P_N(\ell) \ge \sigma  P_N(\ell)$ for some $\sigma >0$. To obtain such a bound, we take the scalar product of the eigenvalue equation with $\mf P_N(\ell) \Psi_N$ and utilize the observation that the two-body potential in $H_N$ acts, after subtraction of the mean-field contribution, effectively as a discrete second derivative in $\ell$. To illustrate the simplicity of the idea, we provide a more detailed sketch of the argument in Section \ref{sec:outline}.

\section{Proof}

The remainder of this note is organized as follows. In the next section, we introduce the excitation formalism \cite{LNSS,LNS}, which is convenient for our proof. In Section \ref{eq:exc:hamiltonian}, we give a heuristic discussion of the proof and in Section \ref{sec:proof}, we state the main technical lemma and use this lemma to prove our main result. In Sections \ref{sect:proof:diff:inequality} and \ref{sec:extension:singular}, we provide the proof of the technical lemma.

\subsection{Excitation Hamiltonian}\label{eq:exc:hamiltonian}

We define the Hartree energy as $e_{\rm H} : = N^{-1} \inf_u \langle u^{\otimes N} , H_N u^{\otimes N} \rangle $, where the infimum is taken over all $L^2$-normalized $u \in H^1(\mathbb R^3)$ and we denote the corresponding unique positive minimizer by $\varphi$. For a proof of existence and uniqueness of the minimizer and its positivity, see e.g. \cite[Lem 2.2]{BPS2021}. Given the Hartree minimizer $\varphi$, we introduce the unitary excitation map $ U_N(\varphi) : \mathfrak H^N \to \mathcal F^{\le N}_{\perp}: = \bigoplus_{\ell=0}^N \bigotimes_{\textnormal{sym}}^\ell \{ \varphi\}^ \perp$ acting as
\begin{align}
U_N(\varphi) \Phi_N = \bigoplus_{\ell=0}^N q^{\otimes \ell} \bigg( \frac{ a(\varphi)^{\otimes N-\ell}}{\sqrt {(N-\ell)!} } \Phi_N \bigg) , \quad \Phi_N\in \mf H^N
\end{align}
where $q=1-|\varphi \rangle \langle \varphi |$. We denote the usual bosonic creation and annihilation operators by $a^*$ and $a$ and the number operator by $\mathcal N $. 

We then define the excitation Hamiltonian $\mb H$ as an operator acting on the excitation Fock space $\mathcal F^{\le N}_{\perp}$ by
\begin{align}\label{eq:def:excitation:H}
\mb H & := U_N(\varphi) (H_N - N e_{\rm H}) U_N(\varphi )^* \notag\\[1mm]
& = \mb K_0 +\frac{1}{N-1} \Big( \mb K_1 \, \mf a (\mc N) + \big( \mb K_2 \, \mf b(\mc N) + \text{h.c.} \big)   + \big ( \mb K_3 \, \mf c(\mc N)    + \text{h.c.} \big) +   \mb K_4 \Big)
\end{align}
with $N$-dependent functions
\begin{align}
\mf a (\l) & := N-\l , \quad \mf b(\l)  :=  \sqrt{(N-\l) (N-\l -1)} , \quad \mf c(\l) := \sqrt{N - \ell}
\end{align}
and $N$-independent operators $\mb K_0 : = \text{d}\Gamma(q h q )$ with $h : L^2(\mathbb R^3) \to L^2(\mathbb R^3)$ given by
\begin{align}\label{eq:definition:T:operator}
h = -\Delta + V^{\textnormal{ext}} + v\ast \varphi^2 - \langle \varphi, ( -\Delta + V^{\textnormal{ext}} + v\ast \varphi^2 ) \varphi \rangle
\end{align}
and\allowdisplaybreaks
\begin{align}
\mb K_1 & := \int \d x_1 \d x_2\, K_1(x_1,x_2) a_{x_1}^* a_{x_2} \\
\mb K_2 & := \frac{1}{2} \int \d x_1 \d x_2 K_2(x_1,x_2) a_{x_1}^* a_{x_2}^* \\
\mb K_3 & := \int dx_1 \d x_2 \d x_3\, K_3(x_1,x_2,x_3) a_{x_1}^* a^*_{x_2} a_{x_3} \\
\mb K_4 & := \frac{1}{2} \int \d x_1 \d x_2 \d x_3 dx_4\, K_4(x_1,x_2,x_3,x_4) a_{x_1}^* a_{x_2}^* a_{x_3} a_{x_4} .
\end{align}
With $K(x,y) := \varphi(y) v(x-y) \varphi(x)$ and
\begin{align}
W(x,y) : = v(x-y)  - v * |\varphi|^2 (x) -   v * |\varphi|^2 (y)  + \ls \varphi , v * |\varphi|^2 \varphi \rs ,
\end{align}
the different kernels are given by
\begin{align}
K_1(x_1,x_2) &:= \int \d y_1\d y_2 \, q(x_1,y_1) K(y_1,y_2) q(y_2,x_2), \\
K_2(x_1,x_2) &:= \int \d y_1\d y_2 \, q(x_1,y_1) q(x_2,y_2) K(y_1,y_2), \\
K_3(x_1,x_2,x_3) &:= \int \d y_1\d y_2 \, q(x_1,y_1) q(x_2,y_2) W(y_1,y_2) \varphi(y_1) q(y_2,x_3), \\
K_4(x_1,x_2,x_3,x_4) &:= \int \d y_1\d y_2 \, q(x_1,y_1) q(x_2,y_2) W(y_1,y_2) q(y_1,x_3) q(y_2,x_4),
\end{align}
where $q(x,y)$ is the integral kernel of $q = 1 - |\varphi \rangle \langle \varphi|$. For the derivation of \eqref{eq:def:excitation:H}, we refer to \cite{BPS2021,LNSS}. Before we continue, let us note the important fact that the operator $qhq$ with $h$ defined in \eqref{eq:definition:T:operator} has a spectral gap above zero, that is, $qhq \ge \tau$ for some number $\tau>0$. Consequently, we have $\mb K_0 \ge \tau \mathcal N$ on the excitation Fock space $\mathcal F_\perp := \bigoplus_{\ell=0}^\infty \bigotimes_{\textnormal{sym}}^\ell \{ \varphi\}^\perp$.

Denoting by $\Psi_N \in \mf H^N$ the unique ground state of $H_N$ with ground state energy $E_N = \inf \sigma(H_N)$, we introduce $\chi := U_N(\varphi) \Psi_N$, which satisfies the eigenvalue equation $\mb H \chi = ( E_N - N e_{\rm H} ) \chi$. In terms of $\chi = (\chi^{(\ell)})_{\ell=0}^N$, the probability to find $\ell$ excitations outside the condensate wave function $\varphi$ is given by $P_N(\ell) = \| \chi^{(\ell)} \|^2$.

\subsection{Idea of the proof}\label{sec:outline}

To illustrate the idea of the proof, we demonstrate the argument for the ground state eigenfunction of the quadratic Bogoliubov approximation of $\mb H$. That is, we consider the eigenvalue equation $\mb H_0 \phi = E_0 \phi $ on $\mathcal F_\perp$, where
\begin{align}\label{eq:Bog:Hamiltonian}
\mb H_0 = \mb K_0 + \mb K_1 + \mb K_2 + \mb K_2^\dagger
\end{align}
and $E_0<0$ is the lowest possible eigenvalue of $\mb H_0$. Existence and uniqueness of the ground state $\phi \in \mathcal F_\perp$ follow by unitary diagonalization of $\mb H_0$ \cite{LNSS,NNS}. Similarly as for $\chi$, the number of particles in $\phi$ correspond to the number of particles in the state $U_N(\varphi)^* 1(\mathcal N \le N) \phi \in \mf H^N$ that are not in the condensate wave function. In the following, we show that $\| \phi^{(\ell)}\|^2 \le C \exp(-\varepsilon \ell)$ for some $C,\varepsilon >0$ and all $\ell\ge 0$. Note that for the purpose of this demonstration, we assume that $v$ is pointwise bounded with $\| v \| _\infty $ sufficiently small. 

We start by taking the scalar product on both sides of  the eigenvalue equation with $\phi^{(\ell)}$. Using $E_0 \le 0$ and $\mb K_0 \ge \tau \mathcal N$ with $\tau > 0$, this implies
\begin{align}
\tau \ell \| \phi^{(\ell)} \|^2 \le - \Ls \phi^{(\el)} \mb K_2 \phi^{(\el-2) }\Rs -  \Ls \phi^{(\el)} \mb K_2^\dagger \phi^{(\el+2) }\Rs -   \Ls \phi^{(\el)} \mb K_1 \phi^{(\el) }\Rs.
\end{align}
Since $\mb K_1\ge 0$ (this follows from $\hat v\ge 0$), we can apply Lemma \ref{lem:K2:bound:b} below to bound the terms involving $\mb K_2$. This leads to
\begin{align}\label{eq:outline:ineuqlaity}
& \bigg( 4 \tau - \frac{ v(0) }{\ell } \bigg) \ell \| \phi^{(\ell)} \|^2 \notag\\
& \quad \le   \Ls \phi^{(\el+2)} \mb K_1 \phi^{(\el+2) }\Rs +  \Ls \phi^{(\el-2)} \mb K_1 \phi^{(\el-2) }\Rs + v(0) \|\phi^{(\ell-2)}\|^2    - 2 \Ls \phi^{(\el)} \mb K_1 \phi^{(\el) }\Rs 
\end{align}
Next, we use $\mb K_1\ge 0$ to estimate the last term and apply Lemma \ref{lem:K1:bound} below to estimate the first two terms on the right-hand-side. Abbreviating $f(\ell) := \ell \| \phi^{(\ell)}\|^2$, we obtain
\begin{align}
\bigg( 4 \tau - \frac{ v(0) }{ \ell } \bigg) f(\ell) \le C \| v\|_\infty \Big(  f(\ell+2) + f(\ell-2) \Big).
\end{align}
Dividing both sides by $C \| v\|_\infty$, the pre-factor on the left side is strictly larger than two if $ \| v\|_\infty$ is sufficiently small. (Note that the spectral gap $\tau>0$ is uniform in $\| v\|_\infty \to 0$.) Thus, we arrive at 
\begin{align}\label{eq:diff:ineq:Bog}
\sigma  f(\ell) \le   f(\ell+2) + f(\ell-2) \quad 
\end{align}
for some $\sigma>2$ and all $\ell\ge 2 $. We now consider $f(\ell)$ separately for $\ell$ even/odd. The difference inequality states that the second discrete derivative of $f(\ell)$ is bounded from below by $(\sigma-2) f(\ell)$. On the one hand, this shows that $f(\ell)$ is convex, and thus has at most one minimum $f(\ell_0)$. On the other hand, the inequality implies that $f(\ell) \le (\sigma - 2 )^{\ell} f(1)$ for $1\le \ell \le \ell_0$ 
and $f(\ell) \ge (\sigma -2 )^{\ell-\ell_0} f(\ell_0)$ for $\ell \ge \ell_0$. By normalization of $\phi$, i.e. $\sum_{\ell=0}^\infty \| \phi^{(\ell)} \|^2 =1$, and since $\sigma >2$, this implies that $f(\ell)$ has no minimum. Consequently,  $\ell \| \phi^{(\ell)} \|^2 \le (\sigma - 2 )^{\ell} ( \|\phi^{(2)} \|^2 +  \|\phi^{(3)} \|^2) $ for all $\ell$, as claimed.

In the next section, we extend the above argument to the ground state $\chi = U_N(\varphi) \Psi_N$ of the excitation Hamiltonian $\mb H$ and remove the assumption that $\| v \|_\infty $ is small. This requires some technical modifications: Most importantly, in \eqref{eq:outline:ineuqlaity} we will not estimate the last term by $-\mb K_1 \le 0$. Instead, we sum both sides over $\ell - L , \ldots , \ell +L$ for some large but fixed $L$. While on the right side, many terms cancel each other, the left-hand side is effectively increased by a factor proportional to $L$. This will help us to remove the smallness assumption on $\|  v \|_\infty$. Due to the presence of $v(0)$, the Coulomb case requires another approximation argument that will be explained in Section \ref{sec:extension:singular}. An obstacle in considering the full Hamiltonian $\mb H$ compared to $\mb H_0$ is the presence of $\mb K_3$ and $\mb K_4$ (for the Coulomb potential, $\mb K_4$ is not relevant since it is non-negative). In order to treat these operators as perturbations, we restrict the derivation of the difference inequality to values $\ell \le \delta N$ for some small $\delta$. This is helpful because $\mb K_3$ and $\mb K_4$ have more than two creation and annihilation operators and additional factors of $(N-1)^{-1/2}$. Having established the exponential decay up to $\ell = \delta N$, it will follow as a simple consequence of the eigenvalue equation that $\| \chi^{(\ell)} \|$ is bounded by $\exp(- \varepsilon N )$ for all $\delta N \le \ell \le N$ and some $\varepsilon >0$.  


\subsection{Difference inequality and proof of Theorem \ref{theorem}}

\label{sec:proof}

The following lemma is the main technical ingredient for the proof of Theorem \ref{theorem}. It provides a suitable generalization of the difference inequality \eqref{eq:diff:ineq:Bog} to the ground state $\chi \in \mathcal F_\perp^{\le N}$ of $\mb H$.

\begin{lemma}\label{lem:FL:inequality} There exist constants $L\ge 1 $, $\sigma>2$ and $\kappa \in (0,1)$ such that the discrete function $F_L(\ell):=\sum_{k=\ell - L}^{\ell+L}  k \| \chi^{(k)}\|^2$ satisfies 
\begin{align}\label{eq:FL:inequality:1}
\sigma F_L(\el) \le F_L(\el +L ) + F_L(\el - L) 
\end{align}
for all $L \le \ell \le \kappa N $ and $N$ large enough.
\end{lemma}
Before we embark on the proof of the lemma, we explore its consequences and derive Theorem \ref{theorem}. 
 
\begin{proof}[Proof of Theorem \ref{theorem}]
We apply \eqref{eq:FL:inequality:1} for $\ell =L , 2 L , 3 L , \ldots, n L$ with $n \le  \kappa  N/L $. In other words, we use that the second discrete derivative of the function $G(\l) := F_L( \el L)$ is bounded from below by $(\sigma-2) G(\ell)$,
\begin{align}
(\sigma -2) G(\el) \le G(\el +1 ) + G(\el -1)  - 2 G(\el)
\end{align}
for all $\ell \in \{ 1 ,\ldots, \kappa' N \} $, where $\kappa' = \kappa / L$. This implies that $G$ is convex and thus attains a unique minimum at $\l_0 \in \mb N \cup \{ \infty\}$. The Inequality further implies that $G(\el)$ is exponentially decaying for $\ell \le \ell_0$ and exponentially increasing for $\ell \ge \ell_0$,
\begin{align}
G(\l) & \le \frac{G(1)}{(\sigma-1)^{\l-1}} \qquad \qquad \forall\ 1 \le \ell \le \l_0 , \\[1mm]
G(k ) & \ge   (\sigma- 1 )^{k-\el} G(\el ) \quad \quad \, \forall\ \l_0  \le \el \le  k \le  \kappa' N ,\label{eq:G:lower:bound}
\end{align}
where the second case is relevant only if $\ell_0 \le \kappa' N $. Using the second bound, we can estimate $G(\ell)$ for $\el \le \kappa' N/2 $: In fact, choosing $\el = \kappa' N /2 $, $k =\kappa' N$ and since $G\le N $ by normalization of $\chi$, \eqref{eq:G:lower:bound} implies $G(\kappa' N/2) \le N (\sigma-1)^{-\kappa'  N/2}$ and thus $G(\ell ) \le N  (\sigma-1)^{- \kappa'  N / 2 } $ for $\ell_0  \le \el \le \kappa' N /2 $. 

From the above, we conclude that  $G(\ell) \le C \exp(-\varepsilon' \ell)$ for for some constants $C,\varepsilon' >0$ and all $1\le \ell \le \kappa' N / 2 $. Since $G(\ell)  = \sum_{k = (\ell -1) L}^{(\ell+1)L } \ell \| \chi^{(\ell)} \|^2$, we obtain
\begin{align}\label{eq:exp:decay}
\sup \big\{ \|\chi^{(k )}\|^2 : (\ell -1 ) L \le k \le (\ell+1) L  \big\} \le G(\ell) \le C \exp(-\varepsilon' \ell)
\end{align}
for all $\ell \le \kappa' N/2$, which implies $\|\chi^{(\ell)}\| \le C \exp(- \varepsilon \ell)$ for all $\ell \le \kappa' N/2$ and some $\varepsilon >0$.

It remains to prove the exponential decay for $\ell \ge \kappa' N/2$. To this end, we write the ground state as $\chi= \chi^\le + \chi^>$ with $\chi^\le = \mathds{1}(\mathcal N\le N\kappa' /2) \chi$. Then, we use the eigenvalue equation together with $\mb H- ( E_N - N e_{\rm H}) \ge 0$ and the fact that only $\mb K_2$ and $\mb K_3$ couple the two different parts of the ground state,
\begin{align}
0 & = \Ls \chi, ( \mb H - ( E_N - Ne_{\rm H}) ) \chi \Rs \notag\\
& \ge   \Ls \chi^>, ( \mb H - ( E_N - N e_{\rm H}) ) \chi^> \Rs  + 2 \text{Re} \Ls \chi^>, \mb K_2 \chi^\le \Rs  + 2 \text{Re} \Ls \chi^>, \mb K_3 \chi^\le \Rs .
\end{align}
The last two terms are bounded by
\begin{align}
\big| \Ls \chi^>, \mb K_2 \chi^\le \Rs \big| + \big| \Ls \chi^>, \mb K_3 \chi^\le \Rs \big| \le C N \big( \| \chi^{(N\kappa' /2-1)} \| +  \| \chi^{(N\kappa' /2)} \| \big) \le \exp(- c  N )
\end{align}
for some $c>0$ and large enough $N$, where we used the exponential decay of $\| \chi^{(\ell)}\|$ for $\ell \le \kappa' N /2$ in the second bound. (Note that the first bound follows from  Lemmas \ref{lem:K2:bound:b}, \ref{lem:bound:K3} and \ref{lem:K2:bound}). Thus,
\begin{align}
 \Ls \chi^>, ( \mb H - (E_N - N e_{\rm H} ) ) \chi^> \Rs  \le  \exp(- c  N ).
\end{align}
Since $\chi^> / \| \chi^> \|$ is not a low-energy state (using $\mb K_0 \ge \tau \mathcal N$, it is not difficult to see that $ \ls \chi^> , ( \mb H - (E_N -N e_{\rm H})   \chi^> \rs \ge c N \| \chi^> \|^2 $ for some $c>0$), we have $\| \chi^> \| \le C \exp(- \varepsilon N )$ for some $C,\varepsilon>0$.

To summarize, there exist constants $C,\varepsilon >0$ such that $\| \chi^{(\ell)} \| \le C \exp(- \varepsilon \ell)$ for all $\ell \in \{ 1,\ldots, N\}$ and all large $N$. This implies Theorem \ref{theorem}.
\end{proof}

\subsection{Proof of the difference inequality}

\label{sect:proof:diff:inequality}

Let us recall that we assume the pair potential $v$ to be either (i) even, pointwise bounded and with non-negative Fourier transform or (ii) $v(x) = \lambda |x|^{-1}$ with $\lambda >0$.  For better readability, we first prove Lemma \ref{lem:FL:inequality} in case (i). In Section \ref{sec:extension:singular}, we explain how the proof is adapted to cover case (ii). Note that in both cases, we have $\| v^2 \ast \varphi^2\|_\infty < \infty$, where $\varphi$ is the normalized positive Hartree minimizer. For the Coulomb potential, this follows from Hardy's inequality and $\varphi \in H^1(\mathbb R^3)$.

We first state and prove some preliminary estimates for the operators appearing in the excitation Hamiltonian. The statements of Lemmas \ref{lem:K1:bound} and \ref{lem:bound:K3} hold in both cases, whereas Lemmas \ref{lem:K2:bound:b} and \ref{lem:bound:K4} only hold in case (i).

\begin{lemma}\label{lem:K1:bound} For all $\xi \in \mathcal F_\perp$, we have
\begin{align}
 \big| \Ls \xi^{(\l)} ,  \mb K_1 \xi^{(\ell)} \Rs  \big|  & \le  \| v^2 \ast \varphi^2\|_\infty^{1/2}\, \ell \|\xi^{(\ell)}\|^2.
\end{align}
\end{lemma}
\begin{proof}
We apply two times Cauchy--Schwarz and use $\mathcal N = \int dx \, a_x^* a_x$ to find
\begin{align}\label{K:1:bound:x}
 \Big| \Ls \xi^{(\l)} , \mb K_{1} \xi^{(\ell)} \Rs  \Big| &  = \Big|  \int \d x \d y  \, \varphi(x) v(x-y) \varphi(y)
\Ls \xi^{(\l)} , a_{x}^* a_{y} \xi^{(\l)} \Rs  \Big|\notag \\
&  \le \int dx\, \varphi(x) \, \big( v^2 \ast \varphi^2 (x)\big)^{1/2} \| a_x\xi^{(\ell)}\|\, \|\mathcal N^{1/2}\xi^{(\ell)}\| \notag\\[2mm]
& \le \| v^2 \ast \varphi^2\|_\infty^{1/2} \|\varphi\|_2 \|\mathcal N^{1/2}\xi^{(\ell)}\|^2.
\end{align}
\end{proof}

\begin{lemma}\label{lem:K2:bound:b} For pointwise bounded $v$ with $\hat v\ge 0$, we have
\begin{align}\label{eq:first:bound:LemmaK2}
4  \big| \Ls \xi^{(\l)} ,  \mb K_2 \xi^{(\ell-2)} \Rs  \big|  & \le  \Ls \xi^{(\l)}, \mb K_1 \xi^{(\l)} \Rs +  \Ls \xi^{(\l-2)}, \mb K_1 \xi^{(\l-2)} \Rs  + v(0) \| \xi^{(\ell -2)}\|^2  \\[1mm]
4  \big| \Ls \xi^{(\l)} ,  \mb K_2^\dagger \xi^{(\ell+2)} \Rs  \big|  & \le  \Ls \xi^{(\l)}, \mb K_1 \xi^{(\l)} \Rs + \Ls \xi^{(\l+2)}, \mb K_1 \xi^{(\l+2)} \Rs   + v(0) \| \xi^{(\ell)}\|^2  
\end{align}
for all $\xi \in \mathcal F_\perp$.
\end{lemma}
\begin{proof} We estimate
 \begin{align}
& \Big| \Ls \xi^{(\l)} , \mb K_{2} \xi^{(\ell-2)} \Rs  \Big|  = \frac{1}{2} \Big|  \int \d x \d y  
\Ls \xi^{(\l)} , K (x,y) a_{x}^* a_{y}^*  \xi^{(\l-2)} \Rs  \Big|\notag \\
& \quad = \frac{1}{2} \Big|  \int \d k\,    \hat v(k)    
\Big\langle  \int \d x\, \varphi(x)  e^{ikx} a_x \xi^{(\l)} , \int \d y\, \varphi(y) e^{iky} a_{y}^*  \xi^{(\l-2)}   \Big\rangle \Big| \notag \\
& \quad \le  \frac{ 1 }{4}   \int \d k\,    \hat v(k)    \Big(
\Big\| \int \d x\, \varphi(x)  e^{ikx} a_x \xi^{(\l)} \Big\|^2 + \Big\|  \int \d y\, \varphi(y) e^{iky} a_{y}^*  \xi^{(\l-2)}  \Big\|^2\Big) 
\end{align}
and note that  
\begin{align}
 & \int \d k\,    \hat v(k)     
\Big\| \int \d x\, \varphi(x)  e^{ikx} a_x \xi^{(\l)} \Big\|^2 = \Ls \xi^{(\l)}, \mb K_{1} \xi^{(\l)} \Rs
\end{align}
while
\begin{align}
 & \int \d k\,    \hat v(k)     
\Big\| \int \d x\, \varphi(x)  e^{ikx} a_x^* \xi^{(\l-2)} \Big\|^2 = \int dx dy \,  K(x,y) \Ls  \xi^{(\l-2)} , a_x a_y^* \xi^{(\l-2)} \Rs \notag \\
& \qquad \qquad \qquad \qquad = v(0) \int \d x |\varphi(x)|^2\,  \| \xi^{(\l-2)}\|^2 +   \Ls  \xi^{(\l-2)}, \mb K_{1} \xi^{(\l-2)} \Rs ,
\end{align}
where we used $a_x a_y^* = \delta(x-y) + a_y^* a_x$ and $K(x,y) = K(y,x)$.

The second bound of the lemma follows from $\ell \mapsto \ell+2$ and $\mb (\mb K_2^\dagger)^\dagger = \mathbb K_2$.
\end{proof}

\begin{lemma}\label{lem:bound:K3} There is a constant $C>0$ so that for all $\ell \le \delta N$, $\delta\in (0,1)$, we have
\begin{align}
\big| \Ls \xi^{(\l)} ,  \mb K_3  \xi^{(\ell-1)}  \Rs  \big| & \le  (\delta N)^{1/2} \big(C  \ell \|\xi^{(\ell)} \|^2 +   (\ell-1) \| \xi^{(\ell-1 )} \|^2  \big) \\[1mm]
\big| \Ls \xi^{(\l)} ,  \mb K_3^\dagger  \xi^{(\ell+1)} \Rs  \big| & \le (\delta N)^{1/2} \big(C  \ell \|\xi^{(\ell)} \|^2 +   (\ell+1) \| \xi^{(\ell+1 )} \|^2  \big)
\end{align}
for every $\xi \in \mathcal F_\perp$.
\end{lemma}
\begin{proof} Using $\| W^2 \ast \varphi^2\|_\infty \le C$, it follows again by Cauchy--Schwarz that
\begin{align}
\Big| \Ls \xi^{(\l)} ,  \mb K_3  \xi^{(\ell-1)} \Rs  \Big| &  =  \Big|   \int  \d x_2 \, W(x_1,x_2) \varphi(x_1)  \Ls \xi^{(\l)} , a_{x_1}^* a^*_{x_2} a_{x_2}  \xi^{(\ell-1)} \Rs \Big | \notag\\
& \le  C \int dx_2  \big(W^2 \ast \varphi^2(x)\big)^{1/2} \| a_{x_2} \mathcal N^{1/2}  \xi^{(\l)} \| \,  \| a_{x_2}  \xi^{(\ell-1)} \|  \notag\\[1mm]
&  \le  C \ell^{3/2} \|\xi^{(\ell-1)} \|   \,   \|\xi^{(\l)}  \|
\end{align}
and similarly, for the bound involving $\mb K_3^\dagger$.
\end{proof}

\begin{lemma} \label{lem:bound:K4} For $\| v \|_\infty \le C $, we have for $\ell \le \delta N$, $\delta\in (0,1)$,
\begin{align}
\big| \Ls \xi^{(\l)} ,  \mb K_4  \xi^{(\ell)}  \Rs  \big| & \le C \delta N  \ell \|\xi^{(\ell)} \|^2
\end{align}
for every $\xi \in \mathcal F_\perp$.
\end{lemma}
The proof is straightforward and thus omitted.

\begin{proof}[Proof of Lemma \ref{lem:FL:inequality}] We shall prove the lemma in two steps.\medskip

\noindent \textbf{Step 1.} We take the scalar product with the state $\chi^{(\ell)}$ on both sides of the eigenvalue equation $\mb H \chi = (E_N -Ne_{\rm H} ) \chi$.  Using $E_N - Ne_{\rm H} \le 0$, $\mathcal N \chi^{(\ell)} = \ell \chi^{(\ell)}$ and multiplying both sides by $N-1$, we obtain
\begin{align}\label{eq:bound:K4}
0\ge &  (N -1) \Ls  \chi^{(\ell)} , \mb K_0  \chi^{(\ell)} \Rs  + \mf a (\ell) \Ls \chi^{(\ell)}, \mb K_1 \chi^{(\ell)} \Rs   \notag \\[1mm]
& \quad + \mf b (\ell-2)  \Ls \chi^{(\ell)} , \mb K_2 \chi^{(\ell-2)} \Rs +  \mf b (\ell) \Ls \chi^{(\ell)} ,  \mb K_2^\dagger \chi^{(\ell+2)} \Rs \notag  \\[1mm]
& \quad + \mf c (\ell-1) \Ls \chi^{(\ell)} ,  \mb K_3 \chi^{(\ell-1)} \Rs + \mf c (\ell) \Ls \chi^{(\ell)},  \mb K_3^\dagger \chi^{(\ell+1)} \Rs +   \Ls \chi^{(\ell)} \mb K_4 \chi^{(\ell)} \Rs 
 \end{align}
and invoking $\mb K_0\ge \tau \mc N$, we find
\begin{align}\label{ineq:proof:abbr}
& (N-1) \tau \ell \| \chi^{(\ell)} \|^2 +  \mf a(\ell) \Ls \chi^{(\ell)}, \mb K_1 \chi^{(\ell)} \Rs \notag\\ 
&\qquad \qquad \qquad \qquad  \le - \mf b (\ell-2)  \Ls \chi^{(\ell)}, \mb K_2 \chi^{(\ell-2)} \Rs -  \mf b (\ell) \Ls \chi^{(\ell)},  \mb K_2^\dagger \chi^{(\ell+2)} \Rs \notag \\
&  \qquad \qquad \qquad \qquad \quad - \mf c (\ell-1) \Ls \chi^{(\ell)},  \mb K_3 \chi^{(\ell-1)} \Rs - \mf c (\ell) \Ls \chi^{(\ell)},  \mb K_3^\dagger \chi^{(\ell+1)} \Rs \notag \\
&  \qquad \qquad \qquad \qquad \quad -  \Ls \chi^{(\ell)} , \mb K_4 \chi^{(\ell)} \Rs.
\end{align}
To facilitate the reading, let us abbreviate $f(\ell) := \ell \| \chi^{(\ell)} \|^2$, $ g(\l) := \Ls \chi^{(\ell)}, \mb K_1 \chi^{(\ell)} \Rs $ and
\begin{align}
R_2(\l) & : = - \mf b (\ell-2)  \Ls \chi^{(\ell)}, \mb K_2 \chi^{(\ell-2)} \Rs -  \mf b(\ell) \Ls \chi^{(\ell)} ,  \mb K_2^\dagger \chi^{(\ell+2)} \Rs \\
R_3(\l) &  := - \mf c (\ell-1) \Ls \chi^{(\ell)},  \mb K_3 \chi^{(\ell-1)} \Rs - \mf c (\ell) \Ls \chi^{(\ell)} ,  \mb K_3^\dagger \chi^{(\ell+1)} \Rs \\
R_4(\ell) & := -  \Ls \chi^{(\ell)} , \mb K_4 \chi^{(\ell)} \Rs.
\end{align}
such that Inequality \eqref{ineq:proof:abbr} reads
\begin{align}\label{ineq:proof:abbr:2}
(N-1) \tau f(\ell)  + \mf a(\l) g(\l) \le  R_2(\l) + R_3(\l) + R_4(\l) .
\end{align}
Since the left-hand side of \eqref{ineq:proof:abbr:2} is non-negative, we can apply Lemma \ref{lem:K2:bound} to estimate
\begin{align}\label{eq:bound:R_2}
| R_2(\l) |  & \le \tfrac{1}{4} \Big( \mf b(\ell-2) g(\l) + \mf b(\l-2) g(\l-2) + \mf b(\l) g(\l) + \mf b(\l) g(\l+2) \Big)  \notag\\
& \quad + C N\ell^{-1}  f(\ell) + C \| \chi^{(\ell-2)} \|^2 
\end{align}
where we used that $\mf b(\l-2), \mf b(\l) \le  N  $ and $ v(0) \le C$. Note that for $\ell -2 \ge \delta^{-1/2}$ for some $\delta \in (0,1)$, we can bound the last term by $ \| \chi^{(\ell-2)} \|^2 \le \delta^{1/2} f(\ell-2)$. If we further restrict the values of $\ell$ to $ \ell \le \delta N $, we have by Lemma \ref{lem:bound:K3} that
\begin{align}
| R_3(\l) |  & \le   \sqrt{ N}  \Big( \big|  \Ls \chi^{(\ell)},  \mb K_3 \chi^{(\ell-1)} \Rs \big| + \big|  \Ls \chi^{(\ell)},  \mb K_3 \chi^{(\ell+1)} \Rs \big| \Big) \notag\\[1mm]
& \le C N \delta^{1/2} f(\ell) + N\delta^{1/2}  f(\ell-1) + N\delta^{1/2}  f(\ell+1)
\end{align}
where we used $\mf c(\l-1 )  \le \mf c(\l) \le \sqrt N $. Moreover, by Lemma \ref{lem:bound:K4}, $|R_4(\ell)| \le C N \delta f(\ell)$. Thus, we arrive at
\begin{align}
4 (N-1)\big(  \tau   - C \ell^{-1} - C \delta^{1/2} \big)  f(\ell) & \le  \mf b(\ell-2) g(\l) +  \mf b(\l-2)  g(\l-2)  - 2 \mf a(\l) g(\l) \notag\\[1mm]
& \hspace{-1cm} + \mf b(\l)  g(\l) +  \mf b(\l)   g(\l+2)   - 2 \mf a(\l) g(\l) \notag\\[1mm]
& \hspace{-1cm } +  CN \delta^{1/2} \big( f(\ell+1) +   f(\ell-1)  + f(\ell-2) \big) . 
\end{align}
We now choose $\ell\ge c $ large enough and $\delta$ sufficiently small so that the left-hand side is bounded from below by $2 N\tau f(\ell) $. Moreover, we write the first two lines of the right-hand side as
\begin{align}
& \mf b(\ell-2) g(\l-2) - \mf b(\ell) g(\ell) + \mf b(\ell) g(\ell+2) - \mf b(\ell-2) g(\ell) \notag\\
& \qquad \qquad \qquad  + 2 ( \mf b(\ell) - \mf a(\ell) ) g(\l) +  2 (\mf b(\l-2) -\mf a(\ell) ) g(\l) 
\end{align}
and use that
\begin{align}
\quad  \mf b(\l) - \mf a(\l)  \le 0, \quad \mf b(\l-2)  - \mf a(\l) \le C , \quad g(\l) & \le C f(\ell)  ,
\end{align}
where the last bound follows from Lemma \ref{lem:K1:bound}. Thus, we obtain the inequality
\begin{align}
N  \tau f(\ell)  & \le  \mf b(\l-2)  g(\l-2) - \mf b(\ell) g(\l)  +  \mf b(\l)  g(\l+2) - \mf b(\ell-2) g(\l)  \notag\\[1mm]
& \quad  + CN \delta^{1/2} \big(  f(\ell+1) + N f(\ell-1) + f(\ell-2) \big).
\end{align}
Now, we sum both sides over $\{ L -\l , \ldots, L+\ell\}$. On the left-hand side, this gives
\begin{align}
N \tau F_L(\l) := N \tau \sum_{k=\l  - L}^{\l + L }    k \| \chi^{(k)}  \|^2, 
\end{align}
whereas the terms on the right-hand side are bounded by
\begin{align}
\sum_{k=\l - L  }^{\l + L  } \big(  \mf b( k -2) g( k -2) - \mf b(k) g(k )  \big) &  \le   \mf b(\l  - L - 2 ) g(\l- L - 2) \notag\\
& \le C N f(\el- L - 2)
\end{align}
and
\begin{align}
\sum_{k=\l-L}^{L+\l} \big( \mf b(k)  g(k+2) - \mf b(k-2) g(k) \big)  & \le  \mf b(\l + L) g(\l + L +2 )   \le C N f(\l + L +2)
\end{align}
and
\begin{align}
& \sum_{k=\ell - L } ^{\ell + L }    N \delta^{1/2} \big(  f(\ell +1 ) +   f(\ell-1) + f(\ell-2) \big) \notag\\[0mm]
& \hspace{1cm} \le  3 N \delta^{1/2} F_L(\ell) +  N \big(   f( \ell - L - 1 )  +  f( \ell + L  + 1 ) + f(\ell -L-2)\big) ,
\end{align}
where we used $\delta^{1/2} \le 1/2 $. 

Putting everything together, we arrive at the conclusion that there is a constant $0 < \mu \le (CN)^{-1}N ( \tau - C \delta^{1/2} ) $ such that for all allowed values of $\l$, that is, for $2 + \delta^{-1/2} \le \ell \le \delta N$ for sufficiently small $\delta$ and all large $N$, we have 
\begin{align}\label{eq:FL:bound}
\mu F_L(\l) & \le  f(\l + L+2 ) + f( \ell + L  + 1 )  +  f( \ell - L - 1 ) +  f(\l -  L -2  ) .
\end{align}

\noindent \textbf{Step 2.} We proceed by estimating
\begin{align}
F_L(\el +L ) + F_L(\el - L) 
& \ge  \sum_{k = \ell + L +1 }^{\l + 2L} f(k) \,  + \, \sum_{k= \el -2 L }^{\ell - L -1  }   f(k ) \notag \\ 
& \ge \mu \Big( F_L(\ell) + F_{L+2}(\el) + F_{L+4}(\ell) + \ldots +F_{2L-2} ( \ell ) \Big)
\end{align}
where we used the definition of $F_L(\ell)$ in the first step and applied Inequality \eqref{eq:FL:bound} with $L\to L+j$ in the second step, that is, 
\begin{align}
& f( \el + L+j +2   ) + f( \el + L+ j  +1 ) \notag\\
& \hspace{2cm} + f(\el - L - j -2  ) + f (\ell - L - j - 1  )  \ge \mu F_{L+j}(\el) 
\end{align}
for $j=0,\ldots,L-2$.

Finally, we invoke $F_{L+j}(\el) \ge F_L(\el)$ to arrive at the desired inequality
\begin{align}\label{eq:FL:inequality}
F_L(\el +L ) + F_L(\el - L) \ge \frac{\mu}{2} (L-3) F_L(\el) =: \sigma F_L(\el).
\end{align}
By choosing $L$ large enough, we have $\sigma >2$, which completes the proof of Lemma \ref{lem:FL:inequality}
\end{proof}

\subsection{Extension to the repulsive Coulomb potential}
\label{sec:extension:singular}

We briefly explain how the proof of Lemma \ref{lem:FL:inequality} presented in the previous section can be adapted to cover the repulsive Coulomb potential $v(x) = \lambda |x|^{-1}$ with $\lambda>0$. Since $\| v^2 \ast \varphi^2 \|_\infty <\infty$,  Lemmas \ref{lem:K1:bound} and \ref{lem:bound:K3} still apply. Lemma \ref{lem:bound:K4}, on the other hand, is not needed, since we can use $\mb K_4 \ge 0$ in \eqref{eq:bound:K4}. (Note that we have not assumed positivity of $v$ in the bounded case). The only obstacle comes from the use of Lemma \ref{lem:K2:bound:b}, which requires $v(0) < \infty$. For the repulsive Coulomb potential, we thus replace Lemma \ref{lem:K2:bound:b} by the following statement. 
\begin{lemma}\label{lem:K2:bound} Let $v(x)=\lambda|x|^{-1}$ with $\lambda>0$. There is a constant $C>0$ such that for every $\varepsilon >0$ there exists a constant $\nu( \varepsilon ) > 0$, such that
\begin{align}\label{eq:first:bound:LemmaK2}
4  \big| \Ls \xi^{(\l)} ,  \mb K_2 \xi^{(\ell-2)} \Rs  \big|  & \le g(\ell) + g(\ell-2) + \nu(\varepsilon) \| \xi^{(\ell-2)}\|^2 + \varepsilon \big( f(\ell-2) +  f(\ell) \big) 
\end{align}
for all $\xi \in \mathcal F_\perp$, where $g(\ell) = \ls \xi^{(\ell)}, \mb K_1 \xi^{(\ell)}\rs$ and $f(\ell) = \ell \| \xi^{(\ell)} \|^2$.
\end{lemma}
The bound for $| \ls \xi^{(\ell)} , \mb K_2^\dagger \xi^{(\ell+2)} \rs |$ is obtained by $\ell \mapsto \ell+2$ and $(\mb K_2^\dagger )^\dagger = \mb K_2$. After invoking Lemma \ref{lem:K2:bound} to bound $|R_2(\ell)|$ in \eqref{eq:bound:R_2}, the crucial point is that we can choose $\varepsilon$ as small as we want (but always fixed w.r.t. $N$), say $\varepsilon \le \delta^{1/2}$. This comes at the cost of a large factor $\nu(\varepsilon)$, which can be compensated by restricting the values of $\ell$ to $\ell \ge 2 +  \delta^{-1/2} \nu(\varepsilon)$. This way, we can bound the third term in \eqref{eq:first:bound:LemmaK2} by $\nu(\varepsilon) \| \xi^{(\ell-2)}\|^2 = \tfrac{\nu(\varepsilon)}{ \ell-2} f(\ell-2)\le \delta^{1/2} f(\ell-2)$. With this at hand, the remaining steps of the proof are completely analogous to the bounded case.

\begin{proof}[Proof of Lemma \ref{lem:K2:bound}] We write $v = v_\kappa + v^\perp_\kappa$ with $v_\kappa(x) = v(x) (1-\exp(- |x| /\kappa ))$, $\kappa>0$, and observe that $v_\kappa(0) = \lambda \kappa^{-1}$ and $\widehat {v_\kappa} \ge 0$. The non-negativity of the Fourier transform follows from the fact that the Fourier transform of the Yukawa potential $ x \mapsto v(x) \exp(- |x| /\kappa )$ is smaller than the Fourier transform of the Coulomb potential. Moreover, we have $\| (v^\perp_\kappa )^2 \ast \varphi^2 \|_\infty \to 0$ as $\kappa \to 0$, which is a consequence of $f_\kappa(x):=  (v^\perp_\kappa )^2 \ast \varphi^2 (x)$ being strictly monotone decreasing as $\kappa \to 0$, i.e., $f_\kappa(x) - f_{\eta}(x) = \lambda \int dy\, |x-y|^{-2} (e^{-2|x-y|/\kappa} - e^{-2 |x-y|/\eta}) \varphi(y)^2 >0$ for all $x\in \mb R^3$ and $\kappa > \eta$, where we used positivity $\varphi$. In analogy to the definitions in Section \ref{eq:exc:hamiltonian}, we define $K_{\kappa}(x,y) = \varphi(x) v_\kappa (x-y) \varphi(y)$ as well as $\mb K_{2,\kappa}$ and $\mb K_{1,\kappa}$. For the part involving $v_\kappa$, we can proceed as in the proof of Lemma \ref{lem:K2:bound:b}, which gives
\begin{align}
4  \big| \Ls \xi^{(\l)} ,  \mb K_{2,\kappa} \xi^{(\ell-2)} \Rs  \big|  & \le  \Ls \xi^{(\l)}, \mb K_{1,\kappa} \xi^{(\l)} \Rs +  \Ls \xi^{(\l-2)}, \mb K_{1,\kappa} \xi^{(\l-2)} \Rs  + \lambda \kappa^{-1} \| \xi^{(\ell-2)}\|^2.
\end{align}
Applying Lemnma \ref{lem:K1:bound} for $v_\kappa^\perp = v - v_\kappa$ and using $\| (v^\perp_\kappa )^2 \ast \varphi^2 \|_\infty \to 0$ as $\kappa \to 0$, we further have
\begin{align}
\big| \Ls  \xi^{(\l)}, ( \mb K_{1,\kappa} - \mb K_{1} ) \xi^{(\l)} \Rs  \big| \le  \| (v_\kappa^\perp)^2 \ast\varphi^2 \|^{1/2}_\infty \ell \| \xi^{(\ell)} \|^2 \le \delta(\kappa) \ell \| \xi^{(\ell)} \|^2
\end{align}
for some sequence $\delta(\kappa) \to 0 $ as $\kappa \to 0$. To estimate the remainder term, we apply two times Cauchy--Schwarz (similarly as in the proof of Lemma \ref{lem:K1:bound}) to obtain
\begin{align}
& \Big| \Ls \xi^{(\l)} , ( \mb K_{2} - \mb K_{2,\kappa}) \xi^{(\ell-2)} \Rs  \Big| = \bigg| \int dx dy \varphi(x) v^\perp_\kappa (x-y)   \varphi(y) \Ls \chi^{(\ell)} , a_x^* a_y^* \chi^{(\ell-2)} \Rs \bigg| \notag\\
& \hspace{2cm} \le C \| (v^\perp_\kappa)^2 \ast \varphi^2 \|_\infty^{1/2} \,  \ell \| \xi^{(\ell)}\| \, \| \xi^{(\ell-2)}\|  \le  \delta(\kappa) \ell \| \xi^{(\ell)}\| \, \| \xi^{(\ell-2)}\| .
\end{align}
Putting both estimates together implies the statement of the lemma.
\end{proof}
\vspace{3mm}

\noindent\textbf{Acknowledgements.} We thank Lea Bo\ss mann, P.T. Nam and S. Rademacher for helpful remarks. P.P. acknowledges funding by the Deutsche Forschungsgemeinschaft (DFG, German Research Foundation) - SFB/TRR 352 "Mathematics of Many-Body Quantum Systems and Their Collective Phenomena".

\end{spacing}


\begin{thebibliography}{11}

\bibitem{BBCS19} C. Boccato, C. Brennecke, S. Cenatiempo and B. Schlein. Bogoliubov Theory in the Gross-Pitaevskii Limit. \emph{Acta Mathematica 222, 219--335}. (2019)

\bibitem{BBCS20} C. Boccato, C. Brennecke, S. Cenatiempo and B. Schlein. Optimal rate for Bose--Einstein condensation in the Gross--Pitaevskii regime. \emph{Commun. Math. Phys., 376:1311--1395}. (2020)

\bibitem{BPS2021} L.~Bo{\ss}mann, S.~Petrat and R.~Seiringer. Asymptotic expansion of low-energy excitations for weakly interacting bosons. \emph{Forum Math. Sigma} 9, E28. (2021)

\bibitem{BSS} C. Brennecke, B. Schlein and S. Schraven. Bogoliubov theory for trapped bosons in the Gross--Pitaevskii regime. \emph{Ann. Henri Poincaré} 23, 1583--1658. (2022)

\bibitem{GrechS} P. Grech and R. Seiringer. The excitation spectrum for weakly interacting bosons in a trap. \emph{Commun. Math. Phys.} 322(2): 559--591. (2013)

\bibitem{KRS} K. Kirkpatrick, S. Rademacher and B. Schlein. A large deviation principle for many--body quantum dynamics.
\emph{Ann. Henri Poincaré} 22, 2595--2618. (2021)

\bibitem{LNR} M. Lewin, P.T. Nam and N. Rougerie. The mean-field approximation and the non-linear Schrödinger functional for trapped Bose gases. \emph{Trans. Amer. Math. Soc.} 368, 6131--6157. (2016)

\bibitem{LNS} M. Lewin, P.T. Nam and B. Schlein. Fluctuations around Hartree states in the mean-field regime. \emph{American Journal of Mathematics} 137, 6 1613--1650. (2015)

\bibitem{LNSS} M. Lewin, P.T. Nam, S. Serfaty and J.P. Solovej. Bogoliubov spectrum of interacting Bose gases. \emph{Comm. Pure Appl. Math.} 68, 3, 413--471. (2015)


\bibitem{LS02} E.H. Lieb and R. Seiringer. Proof of Bose--Einstein condensation for dilute trapped gases. \emph{Phys. Rev. Lett.} 88(17):170409. (2002)

\bibitem{LSSY}  E.H. Lieb, R. Seiringer, J.P. Solovej, and J. Yngvason. The Mathematics of the Bose Gas and its Condensation. Birkh\"auser. (2005)

\bibitem{Mitrouskas2017} D.~Mitrouskas. Mean-field equations and their next-order corrections: Bosons and fermions. PhD Thesis, LMU M\"unchen. (2017)


\bibitem{NNR} P.T. Nam, M. Napi\'{o}rkowski, J. Ricaud and A. Triay. Optimal rate of condensation for trapped bosons in the Gross--Pitaevskii regime. \emph{Analysis and PDE} 15(6): 1585-1616. (2022)


\bibitem{NNS} P.T. Nam, M. Napi\'{o}rkowski  and J.P. Solovej. Diagonalization of bosonic quadratic Hamiltonians by Bogoliubov transformations. \emph{Journal of Functional Analysis} 270, 11, 1, 4340--4368. (2016)

\bibitem{NR23} P.T. Nam and S. Rademacher. Exponential bounds of the condensation for dilute Bose gases. Preprint. (2023)


\bibitem{NRS} P.T. Nam, N. Rougerie and R. Seiringer. Ground states of large bosonic systems: The Gross--Pitaevskii limit revisited.  \emph{Analysis and PDE} 9(2):459--485. (2016)

\bibitem{NT}  P.T. Nam and A. Triay. Bogoliubov excitation spectrum of trapped Bose gases in the Gross--Pitaevskii regime. \emph{Journal de Mathématiques Pures et Appliquées} 176, 18-101. (2023)

\bibitem{Rademacher} S. Rademacher. Large deviations for the ground state of weakly interacting Bose gases. Preprint. \href{https://arxiv.org/abs/2301.00430}{\emph{arXiv:2301.00430}}. (2023)

\bibitem{RademacherS} S. Rademacher and R. Seiringer. Large deviation estimates for weakly interacting bosons. \emph{J. Stat. Phys. 188} (9). (2022)

\bibitem{Seiringer} R. Seiringer. The excitation spectrum for weakly interacting bosons. \emph{Commun. Math. Phys.} 306(2):565--578. (2011)


\end{thebibliography}
\end{document}